\RequirePackage{fix-cm}
\documentclass[11pt]{article}
\usepackage[T1]{fontenc}
\usepackage[latin9]{inputenc}
\usepackage{enumitem}
\usepackage{amsmath}
\usepackage{amsthm}
\usepackage{amssymb}
\usepackage{stmaryrd}
\usepackage{setspace}
\usepackage[authoryear]{natbib}
\PassOptionsToPackage{normalem}{ulem}
\usepackage{ulem}
\usepackage{microtype}
\makeatletter
\numberwithin{equation}{section}
\theoremstyle{plain}
\newtheorem{thm}{\protect\theoremname}

\usepackage[a4paper]{geometry}

\usepackage{bbm}

\usepackage{subfigure}

\usepackage{tikz}

\usepackage{lmodern}
\usepackage[T1]{fontenc}

\usepackage[hidelinks]{hyperref}

\renewcommand{\phi}{\varphi}

\theoremstyle{remark}

\newtheorem{prop}{\textnormal{\bfseries Proposition}}

\newtheorem{lemma}{\textnormal{\bfseries Lemma}}

\newtheorem{corr}{\textnormal{\bfseries Corollary}}

\usepackage[multiple]{footmisc}

\makeatother

\usepackage{babel}
\providecommand{\theoremname}{Theorem}

\begin{document}
\title{\textbf{Axiomatic Foundations of Bayesian Persuasion}\thanks{
We would like to thank the audiences at 2025 SAET Conference (Ischia). 
Takeoka gratefully acknowledges financial support from JSPS KAKENHI Grant Number JP25K00619.}}
\author{Youichiro Higashi\thanks{
Faculty of Economics, Okayama University, 3-1-1 Tsushima-naka, Kita-ku, Okayama 700-8530, Japan. 
Email: higash-y@okayama-u.ac.jp}, \hspace{2mm}
Kemal Ozbek\thanks{Department of Economics, University of Southampton, University Road, 
Southampton, S017 1BJ, United Kingdom. Email: mkemalozbek@gmail.com}, \hspace{2mm}
Norio Takeoka\thanks{Department of Economics,
Hitotsubashi University, 2-1 Naka, Kunitachi, Tokyo 186-8601, Japan. 
Email: norio.takeoka@r.hit-u.ac.jp.}}

\maketitle

\begin{abstract}
{\normalsize{} In this paper, we study axiomatic foundations of Bayesian persuasion, where a principal (i.e., sender) delegates the task of choice making after informing a biased agent (i.e., receiver) about the payoff relevant uncertain state (see, e.g., \cite{KG11}). Our characterizations involve novel models of Bayesian persuasion, where the principal can steer the agent's bias after acquiring costly information.  Importantly, we provide an elicitation method using only observable menu-choice data of 
the principal, which shows how to construct the principal's subjective costs of acquiring information even when he anticipates managing the agent's  bias.}{\normalsize\par}
\end{abstract}

\section{Introduction}

Many economic activities involve a principal and an agent interacting in a choice situation. The principal delegates the task of making choices to the agent in the hope that informed bias-free decisions will be made. To reduce the effects of a bias between himself and the agent, the principal commits to a choice set restricting the options the agent can have; moreover, to inform the agent about the payoff relevant state, the principal acquires information after committing to the choice set. This model of decision-making, at least since the seminal work of \cite{KG11}, is known as the Bayesian persuasion model.

In this paper, we study axiomatic foundations of Bayesian persuasion. We consider a set of suitable axioms for the principal's preferences over choice sets (i.e., menus) and show that these axioms characterize a number of Bayesian persuasion models. Importantly, we show how to elicit the principal's subjective costs by using choice data. The elicitation method is constructive and robust to endogenous changes in the agent's biases.  Our elicitation result helps us to provide a comparative statics exercise that offers a behavioral comparison of differing costs. To the best of our knowledge, ours is the first work in the literature that considers the scenario where the agent's bias may be uncertain and, therefore, the principal may have an incentive to manage the agent's bias after acquiring costly information.

Among the behavioral implications we identify, the monotonicity axioms are key for the Bayesian persuasion models. There are many models of menu-choice that consider only "negative states" (e.g., temptations as in \cite{GP01}; internal-conflicts as in \cite{MO18}; delegations as in \cite{kopylov2020revealed}), or that consider only "positive states" (e.g., subjective states as in \cite{DLR01}; exogenous posteriors as in \cite{DLST14}; endogenous posteriors as in \cite{dDMO17}). Unlike these models, the Bayesian persuasion model does not imply only a desire for commitment or a desire for flexibility. It is possible that in some choice situations the principal may prefer smaller sets to limit the losses due to the agent's biases, and in some other choice situations the principal may prefer bigger sets so as not to miss the options better suited with the realized information about the true state.

Our monotonicity axioms, in terms of observable menu-choice data, provide suitable comparisons to discipline the two channels that can affect the value of a menu: potential biases and uncertain states. In particular, our taste-dominance axiom (Axiom $5$) requires that if a constant menu $A$ (containing only state-independent acts for which information has no value) leads to a better choice than another constant menu $B$ regardless of the agent's potential bias, then $A$ should be weakly preferred. Meanwhile, our information-dominance axiom (Axiom $6$) requires that if a regular menu A (containing also state-dependent acts for which information has value) leads to a better choice than another regular menu B regardless of the agent's posterior belief, then A should be weakly preferred. We show that all Bayesian persuasion models we study satisfy both monotonicity axioms, providing a common ground for these models (see Theorems \ref{known}-\ref{sequential}).

In general, the costs of acquiring information or managing bias are subjective and, therefore, they are not observable. We provide an elicitation method which uses only observable menu-choice data to construct these costs. Specifically, we apply our elicitation method using our most general model of Bayesian persuasion, namely the sequential costly Bayesian persuasion model (see Theorem \ref{ident}). The signifying behavioral implication of the sequential model is the exposure axiom (Axiom 8), which requires that for every menu $A$ there exists an information structure such that the relative value of this particular information structure will be the most with the menu $A$. In the sequential model, the principal first acquires costly information knowing that once a posterior belief is realized, he will then exert costly effort to manage the agent's bias. As such, there are two types of costs for the principal; one for acquiring information about the payoff relevant true state and one for managing the bias of the agent. 

There are related studies that utilize a similar elicitation method that we apply, such as \cite{dDMO17}, who identify subjective information costs, and \cite{MO18}, who identify subjective self-regulation costs. Although our identification of the costs of managing bias is close to \cite{MO18}'s method of using constant menus and singleton equivalents, our identification of costs of information acquisition significantly differs from \cite{dDMO17}'s method of using regular menus and singleton equivalents. While \cite{dDMO17} use, under each posterior belief, the value of best act in a given menu, we generalize their method and use,  under each posterior belief, the value of the whole induced menu and not just the value of the best act for the commitment ranking of the principal. Thus, since we allow for strategic interactions of two people once a posterior is determined and not just the strategic rationality of a single person (see \cite{Kreps79}), our elicitation method of subjective information costs is robust to settings even when there is a wedge between the principal and the agent in terms of subjective tastes.

There are several works in related literature that inform our study on Bayesian persuasion. \cite{Jakobsen21} takes as primitive (i) (principal's) preferences over information structures, indexed by menus of acts, and (ii) (agent's) choice correspondences from each menu (indexed by signals) to study the Bayesian persuasion model in a richer choice setting. In another work, \cite{Jakobsen24} considers (principal's) (i) menu preferences, (ii) random choices, and (iii) state-contingent random choices to analyze the Bayesian persuasion model in another rich choice setting. In a more recent work, \cite{Mensch2025} considers (principal's) state-dependent stochastic choice, and provides an axiomatic characterization of the Bayesian persuasion model in his framework. We view our analysis and these contemporary works as complementary. In particular, we offer a new type of Bayesian persuasion model (e.g., the sequential model); it will be interesting to identify the ex-post choice implications of this model using either Jakobsen's or Mensch's framework of stochastic choices.

The remainder of the paper is organized as follows. In section 2, we introduce our framework and formally define our Bayesian persuasion models. In section 3, we list all behavioral implications of Bayesian persuasion and discuss their relation to the specific models. In section 4, we provide our axiomatic analysis including (i) characterizations of four different Bayesian persuasion models, (ii) elicitation of costs of information acquisition and bias management, and (iii) comparative statics results on these costs. In section 5, we discuss some related choice models in detail. In section 6, we conclude. Proofs of all results given in the text are provided in an Appendix.

\section{Framework and model}

In this section, we develop our framework and define the Bayesian persuasion model and its several extensions.

\subsection*{Environment}

Let $X$ be a finite set of $n$ prizes, with typical elements $x,y,z\in X$
called outcomes. 
$\Delta(X)$ denotes the set of all probability distributions
on $X$, with typical elements $a,b,c\in\Delta(X)$ called lotteries. $S$ is a finite set of $k$ states, $s_{1},s_{2},...,s_{k}$, representing the uncertainty.
${\cal P}=\Delta(S)$ denotes the set of all probability distributions
on $S$, with typical elements $p,q,r\in\Delta(S)$ called beliefs.
$F$ denotes the set of all functions $S\to\Delta(X)$, with typical
elements $f,g,h$ called acts. 
$F^{c}$ denotes the set of all constant functions in $F$. 
With some abuse of notation, we sometimes denote typical elements by $a,b,c\in F^{c}$. 
$\mathbb{A}$ denotes the set of non-empty closed subsets of $F$, with typical elements $A,B,C\in\mathcal{\mathbb{A}}$ called menus. 
$\mathbb{A}^{c}$ denotes the set of closed non-empty subsets of $F^{c}$. 
For $f\in F$ and $p\in\Delta(S)$, let $f^{p}\in F^{c}$
be the induced constant act such that $f^{p}(x)=\sum_{s\in S}p(s)f(s)(x)$
for all $s\in S$. 
For $A\in\mathbb{A}$, let $A^{p}=\{f^{p}\in F^{c}:f\in A\}$ denote the induced constant menu under the belief $p\in\cal{P}$.

Our primitive is a binary relation $\succsim$ over the set of menus
$\mathbb{A}$, with an asymmetric part $\succ$ and a symmetric part $\sim$.
$A\succsim(\succ)B$ means menu $A$ is ``weakly (strictly) preferred''
over menu $B$. The principal chooses a menu in an ex-ante period,
anticipating that he will inform and/or manage an agent in an interim
period, before the agent selects an alternative from the menu in an
ex-post period after her taste and belief are realized. We call the
restriction of $\succsim$ to the set of singleton menus \emph{commitment
ranking}, which represents the principal's preferences over singleton
menus. For any $A\in\mathbb{A}$, let $f_A\in F$ denote the singleton equivalent such that $\{f_A\} \sim A$ and let $x_{A}\in F^{c}$ denote the constant equivalent such that $\{x_{A}\}\sim A$. We denote the restriction of $\succsim$ to the set of constant menus by $\succsim^c$. We say a functional $U:\mathbb{A}\to\mathbb{R}$ represents
$\succsim$ when, for all menus $A$ and $B$, $A\succsim B\;\Leftrightarrow\;U(A)\ge U(B)$.

Let $\mathcal{V}=\{v\in\mathbb{R}^{n}:\sum_{i=1}^{n}v_{i}=0 \; \text{and}\;v\cdot v=1\} $,
with typical elements $u,v,w\in\mathcal{V}$ called\emph{ utilities}
representing differing tastes. Let $\Delta({\cal V})$ denote the set of distributions over $\cal{V}$. We endow $\Delta(\cal{V})$ with the weak*-topology. Note that for any non-constant $w\in\mathbb{R}^{n}$,
there exists a unique $v_{w}\in\mathcal{V}$ such that, for all $a,b\in\Delta(X)$, $w(a)\ge w(b)\;\;\;\text{if and only if}\;\;\;\;v_{w}(a)\ge v_{w}(b)$ where $w(a)$ means $w\cdot a$ for all $w\in\mathbb{R}^{n}$ and
$a\in\Delta(X)$. We also write $w(f)$ to indicate a \emph{utility
act} $S\to w(\Delta(X))$ for all $w\in{\cal V}$, where $w(f)_{i}=w(f_{i})$
for all $i=1,...,k$. We say two utilities $u,v\in{\cal V}$ \emph{conflict} on the ranking
of $a,b\in\Delta(X)$ whenever $u(a)>u(b)$ and $v(a)<v(b)$, or vice
versa. Given $p\in P$, $u$ and $v$ conflict in menu $A$ whenever
$m_{u,p}(A)\cap m_{v,p}(A)=\emptyset$ where for any $w\in{\cal V}$,
set $m_{w,p}(A)=\arg\max_{f\in A}w(f)\cdot p$ denotes the acts in
menu $A$ that maximize $w$ under $p$. We will denote by $u\in{\cal V}$
the \emph{utility} representing the principal's commitment ranking,
and let ${\cal V\setminus}\{u\}$ be the set of all utilities conflicting
with $u$ in general, that is, ${\cal V\setminus}\{u\}$ is the set
of all utilities which has a \emph{bias relative to} $u\in{\cal V}$.

Denote the mixture of $A$ and $B$ by $A\alpha B=\left\{ \alpha f+(1-\alpha)g\in F\,:\,f\in A,\:g\in B\right\}$  for any $\alpha\in[0,1]$ where $h=\alpha f+(1-\alpha)g$ denotes a mixture act such that $h(s)=\alpha f(s)+(1-\alpha)g(s)$
for each $s\in S$. Let $\mathrm{co}(A)$ denote the convex hull of $A\in \mathbb{A}$. For any $A\in\mathbb{A}$ and $\tau\in\Delta({\cal P})$, let $x_{A}(\tau)=\int x_{A^{p}}\,\tau(dp)$ denote the $\tau$-mixture of constant equivalent acts $x_{A^p}$ corresponding to $A^p$ for each $p\in\cal{P}$. For any $p\in \cal{P}$, let $\delta_p \in \Delta(\cal{P})$ denote the degenerate measure that places the point mass on $p$. We will denote by $p_{0}\in{\cal P}$ the \emph{prior belief} governing
the principal's initial uncertainty about the true state. Let $\Delta_{0}({\cal P})$
denote the set of distributions $\tau$ over ${\cal P}$ such that
$\int_{\cal{P}}p\,\tau(dp)=p_{0}$; that is, $\Delta_{0}({\cal P})$ represents
possible information structures that can be
obtained in our framework.
Let $\Delta_0(\cal{P}\times\cal{V})$ denote the set of distributions over the joint space $\cal{P}\times\cal{V}$ such that $\int_{\cal{P}}p\,\pi_{\cal{P}}(dp)=p_{0}$ where $\pi_{\cal{P}}$ is the marginal of $\pi\in\Delta(\cal{P}\times\cal{V})$ over $\cal{P}$. We endow both $\Delta_0({\cal P})$ and $\Delta_0({\cal P}\times\mathcal{V})$ with the weak*-topology.

\subsection*{Bayesian persuasion}
We now introduce the Bayesian persuasion models that we axiomatically study.

\paragraph*{Persuasion with known bias:}

Persuasion is a way of controlling the agent's information by inducing
a random realization of beliefs. Following \cite{KG11},
we say that the principal's preferences $\succsim$ over menus can
be represented by a \emph{Bayesian persuasion} (with known bias)
model if for all $A\in\mathbb{A}$, we have 
\[
U(A)=\max_{\tau \in \Gamma}\int_{{\cal P}}b_{A}^{u,v}(p)\,\tau(dp)
\]
where $\Gamma\subset\Delta_{0}({\cal P})$ is some closed convex set and
$b_{A}^{u,v}(p)=\max_{f\in m_{v,p}(A)}u(f)\cdot p$ is the Strotz function with some utility function $v\in{\cal V}$.
In this case, we say $\succsim$ is represented by the tuple $(u,p_{0},\Gamma,v)$. We also call this preference (model) in short, the persuasion with known bias preference (model).

\paragraph*{Persuasion with uncertain bias:}

The agent's bias may not be known by the principal, but
rather the principal may have a belief about the bias. We say that the principal's preferences $\succsim$ over menus can
be represented by a \emph{Bayesian persuasion} (with uncertain bias)
model if for all $A\in\mathbb{A}$, we have 
\[
U(A)=\max_{\tau \in \Gamma}\int_{{\cal P}} b_{A}^{u,\lambda}(p)\tau(dp)
\]
where $\Gamma\subset\Delta_{0}({\cal P})$ is some closed convex  set and $b_{A}^{u,\lambda}(p)=\int_{{\cal V}}\left(\max_{f\in m_{v,p}(A)}u(f)\cdot p\right)\lambda(dv)$
is a random Strotz function and $\lambda$ is a typically non-degenerate
distribution over ${\cal V}$. In this case, we say $\succsim$ is
represented by the tuple $(u,p_{0},\Gamma,\lambda)$. We also call this preference (model) in short, the persuasion with uncertain bias preference (model).

\paragraph*{Costly persuasion:}

The principal may incur costs to optimally choose the information structure
to persuade the agent (with known or uncertain bias). 
Such a generalization of Bayesian persuasion is studied by 
\cite{gentzkow2014costly}. 
We say that the principal's preferences $\succsim$ over menus can
be represented by a \emph{costly Bayesian persuasion} model if for
all $A\in\mathbb{A}$, we have 
\[
U(A)=\max_{\tau\in\Delta_{0}({\cal P})}\left[\int_{{\cal P}} b_{A}^{u,\lambda}(p)\tau(dp)\,-\,c_{{\cal P}}(\tau)\right]
\]
where $c_{{\cal P}}:\Delta_{0}({\cal P})\to[0,\infty]$ is a proper
lower semi-continuous cost function. In this case, we say that $\succsim$
is represented by the tuple $(u,p_{0},c_{{\cal P}},\lambda)$. We sometimes call this preference (model) in short, the costly persuasion preference (model).

\paragraph*{Sequential persuasion:}

In general, after acquiring costly information, the principal can try to manage the agent's bias by incurring costs. We say that the principal's preferences $\succsim$ over menus can be represented by  a \emph{sequential costly Bayesian persuasion} model if for all $A\in\mathbb{A}$,
\[
U(A)=\max_{\tau\in\Delta_{0}({\cal P})}\left[\int_{{\cal P}}\left(\max_{\lambda\in\Delta(\mathcal{V})}[b_{A^{p}}^{u}(\lambda)-c_{{\cal V}}(\lambda)]\right)\tau(dp)\,-\,c_{{\cal P}}(\tau)\right]
\]
 where $b_{A^{p}}^{u}(\lambda)=\int_{{\cal V}}\,[\,\max_{f\in m_{v,p}(A)}u(f)\cdot\, p\,]\,\lambda(dv)$, and both
 $c_{{\cal P}}:\Delta_{0}({\cal P})\to[0,\infty]$
and $c_{{\cal V}}:\Delta({\cal V})\to[0,\infty]$ are lower semi-continuous proper cost functions. In this case, we
say that $\succsim$ is represented by the tuple $(u,p_{0},c_{{\cal P}},c_{{\cal V}})$. We say $c_{{\cal P}}$ and $c_{{\cal V}}$ are monotone if they are increasing in the order of Blackwell informativeness (see, e.g., \cite{dDMO17}) and in the order of stochastic conflicts (see, e.g., \cite{MO18}), respectively.

Note that the sequential costly Bayesian persuasion model can be written more compactly as: for all
$A\in\mathbb{A}$, 
\[
U(A)=\max_{\pi\in\Delta_0({\cal P}\times\mathcal{V})}\left[b_{A}^{u}(\pi)-c(\pi)\right]
\]
where 
\[
b_{A}^{u}(\pi)=\int_{{\cal {\cal P}\times V}}\,\,[\,\max_{f\in m_{v,p}(A)}u(f)\cdot p\,]\,\,\pi(dp\times dv)
\]
for any $\pi\in\Delta_0({\cal P}\times{\cal V})$ and $c:\Delta_{0}({\cal P}\times{\cal V})\to[0,\infty]$ is a separable cost function 
such that $c(\pi)=c_{{\cal P}}(\pi_{{\cal P}})+\int c_{{\cal V}}(\pi_{p})\pi_{{\cal P}}(dp)$
for all $\pi\in\Delta_{0}({\cal P}\times{\cal V})$ where $\pi_{{\cal P}}$
is the marginal distribution over ${\cal P}$ and $\pi_{p}$ is a
conditional distribution over ${\cal V}$ for every $p\in\Delta(S)$. We sometimes call this preference (model) in short, the sequential persuasion preference (model). We say that a cost function, either $c$, or $c_{\cal{P}}$, or $c_{\cal{V}}$, is grounded if it assumes the value $0$ in its effective domain.

\section{Axioms}

In this section, we discuss the axioms that we use for our characterization results. In particular, the persuasion models described above (which are \emph{unobservable})
induce certain properties on the principal's preferences (which
are \emph{observable}). Next, we look at these implications.

\subparagraph*{Standard axioms:}

The following axioms are standard in the menu-choice literature.

\begin{itemize}
\item \textbf{Axiom 1.} (i) $A\succsim B$ or $B\succsim A$, (ii) $A\succsim B$
and $B\succsim C$ implies $A\succsim C$, and (iii) $A\nsim B$ for
some $A,B$.
\end{itemize}

\begin{itemize}
\item \textbf{Axiom 2.} The sets $\left\{ \alpha\in[0,1]\,:\,A\alpha B\succsim C\right\}$ and $\left\{ \alpha\in[0,1]\,:\,C\succsim A\alpha B\right\}$ are closed.
\end{itemize}

\begin{itemize}
\item \textbf{Axiom 3.} If $A\alpha\{f\}\succsim B\alpha\{f\}$, then $A\alpha\{g\}\succsim B\alpha\{g\}$ for all $g\in F$ and $\alpha\in(0,1)$.

\end{itemize}

\begin{itemize}
\item \textbf{Axiom 4.} If $\mathrm{co}(A)=\mathrm{co}(B)$, then $A\sim B$.
\end{itemize}
Axiom $1$ is the usual non-trivial weak order axiom. Axiom $2$
is the mixture-continuity axiom. Axiom $3$ is the singleton-independence axiom reflecting the idea that singleton menus do not affect the hidden actions (i.e., information acquisition or bias management) taken by the principal. Axiom $4$ states that only the extreme alternatives in a choice set are material indicating that both the principal and the agent are expected-utility maximizers.

\subparagraph*{Dominance:}

The following monotonicity axiom is natural for the preferences of a principal concerned with conflicting
utilities. For constant menus $A,B \in \mathbb{A}^c$, let $A\trianglerighteqslant_{\cal{V}} B$ (read $A$ taste-dominates $B$) if  $A_{a}\,1/2\,B_{b}\,\supset\,B_{a}\,1/2\,A_{b}$ for all $a,b\in \Delta(X)$ with $\{a\}\succ\{b\}$, where $C_{d}=\left\{ c\in C:\{c\}\sim\{d\}\right\} $ is the set of lotteries in constant menu $C\in\mathbb{A}^c$ which are indifferent to the lottery $d\in \Delta(X)$.
\begin{itemize}
\item \textbf{Axiom 5.} [Taste Dominance] If $A\trianglerighteqslant_{\cal{V}} B$, then $A\succsim B$.\bigskip{}
\end{itemize}

This axiom reflects the idea that whenever menu $A\in \mathbb{A}^c$, regardless of the agent's bias, leads to a better choice than menu $B\in\mathbb{A}^c$, then $A$ should be deemed as more preferable than $B$. The Taste Dominance axiom first appears in \cite{MO18}, and they discuss it in more detail in the context of single person decision making. 

The following information dominance axiom is more novel to our setting. For menus $A,B \in \mathbb{A}$, let $A\trianglerighteqslant_{\cal{P}} B$ (read $A$ information-dominates $B$) if $A^{p}\succsim B^{p}$ for all $p\in\Delta(S)$.

\begin{itemize}
\item \textbf{Axiom 6.} [Information Dominance] If $A\trianglerighteqslant_{\cal{P}} B$, then $A\succsim B$.
\end{itemize}

Axiom 6 reflects the idea that if menu $A$,  under any posterior belief $p\in \Delta(\cal{P})$, is more preferable than menu $B$, then $A$ should be deemed better than $B$. Thus, learning about the true state is instrumental for the principal.

\subparagraph*{Desire for commitment:}

The principal's ability to take costly actions (acquiring information or managing bias) is revealed by the following axiom.

\begin{itemize}
\item \textbf{Axiom 7.} [Increasing Desire for Commitment] If $A\sim\{f\}\;\text{and}\;B\sim\{g\}$, then $\{f\}\alpha\{g\}\succsim A\alpha B$ for any $\alpha\in(0,1)$.
\end{itemize}

Desire for commitment increases when menus are mixed since \emph{incentives}
for taking costly action change as a consequence of: (i) decrease in utility
gap between better and worse alternatives for each state in a mixed
menu and (ii) resulting in a decrease in the benefits of costly action
(in terms of both acquiring information and managing bias). This axiom first appears in \cite{MO18}, and they discuss it in more detail in the context of single person decision making.

\subparagraph*{Exposure:} 

The principal's ability to sequentially take actions (information acquisition and bias management) is revealed by the following axiom.

\begin{itemize}
 \item \textbf{Axiom 8.} [Exposure] For each $A\in\mathbb{A}$, there exists
$\tau\in\Delta({\cal P})$ such that (i) for all $B\in\mathbb{A}$,
$\{x_A(\tau)\}\,1/2\,\{x_B\}\succsim
\{x_A\}\,1/2\,\{x_B(\tau)\}$ and (ii) for all singleton menus $C,D \in\mathbb{A}$, $\{x_C(\tau)\}\,1/2\,\{x_D\}\sim
\{x_C\}\,1/2\,\{x_D(\tau)\}$.
\end{itemize}

The Exposure axiom implies that for a given menu $A$, there always exists a distribution $\tau$ such that the marginal gain by moving from the menu $A$ to the $\tau$-mixture of it $x_{A}(\tau)$ yields more value than applying the same operation to another menu $B$. Moreover, $\tau$ is valuable only when it is instrumental in making choices. This is absent when the menus are singletons, in which case there is neutral-behavior for comparing marginal gains.

\subparagraph*{Structural axioms:} Independence axioms allow us to understand whether the principal can take hidden actions (e.g., information acquisition or bias management). A form of the independence axiom for preferences over menus can be given by restricting attention to constant menus.

\begin{itemize}
\item \textbf{Axiom 9.} [Constant-menu Independence] For any $A,B,C\in\mathbb{A}^c$
and $\alpha\in(0,1)$, $A\succsim B$ if and only if $A\alpha C\succsim B\alpha C$.
\end{itemize}

Another form of the independence axiom for preferences over menus can be given by considering singleton menus.

\begin{itemize}
\item \textbf{Axiom 10.} [Singleton-menu Independence] For any $A,B\in\mathbb{A}$, $f\in F$, and $\alpha\in(0,1)$, $A\succsim B$ if and only if $A\alpha \{f\}\succsim B\alpha \{f\}$.
\end{itemize}

 If Axiom 10 holds, then we observe that the principal acquires information within a constraint set (see Theorems \ref{known} and \ref{uncertain}), whereas if only Axiom 9 holds, then the principal acquires costly information (see Theorem \ref{costly}). When neither independence axioms hold, this reveals us that the principal not only acquires costly information, but also manages the agent's bias subject to costs (see Theorem \ref{sequential}).

Moreover, we consider an axiom which captures a measure of the principal's uncertainty about the agent's potential biases.

\begin{itemize}
\item \textbf{Axiom 11.} [Reducibility] For any $a,b\in F^{c}$, $\{a,b\}\sim\{a\}$ or $\{a,b\}\sim\{b\}$.
\end{itemize}

The Reducibility axiom reflects the idea that there is a unique utility conflicting with the principal's commitment ranking (see section 5.1 for a related discussion).

\section{Analysis}

In this section, we provide an axiomatic analysis of characterizing the various Bayesian persuasion models listed above starting from the most specific and proceeding towards the most general model. In particular, our results given below show that the preceding axioms characterize the behavior of a principal who chooses among menus ``as if'' he anticipates following a certain type of Bayesian persuasion model. Importantly, we show (at the end of this section) how the costs of acquiring information or managing bias can be constructed by using menu-choice data. We also provide a comparative statics analysis of the costs.

\paragraph{Persuasion with known bias:}
We start our analysis with our most specific Bayesian persuasion model.

\begin{thm} \label{known}
A binary relation on menus $\succsim$ is a Bayesian persuasion (with known bias) preference,
represented by $(u,p_{0},\Gamma,v)$, if and only if it satisfies Axioms
1-11.
\end{thm}

Theorem 1 shows that the principal's anticipation of following the Bayesian persuasion model when interacting with an agent with known bias can be revealed by looking at the menu-choice behavior of the principal. In particular, the set of information structures $\Gamma$ within which the principal anticipates choosing an information structure to persuade the agent can be identified by considering his preferences over menus (see Theorem 5).

\paragraph{Persuasion with uncertain bias:}

The following result shows that we should forgo the Stable Choice axiom in order to allow for bias uncertainty from the perspective of the principal.

\begin{thm} \label{uncertain}
A binary relation on menus $\succsim$ is a Bayesian persuasion with
uncertain bias preference, represented by $(u,p_{0},\Gamma,\lambda)$,
if and only if it satisfies Axioms 1-10.
\end{thm}

Theorem 2 shows that when all axioms of Theorem 1 except the Stable Choice axiom (Axiom 11) are observed, then this characteristic of the principal's preferences can be attributed to his uncertainty, measured by $\lambda$, about the agent's realized bias. The principal still acquires an optimal information structure within the set $\Gamma$ to persuade the agent about the payoff relevant state but is uncertain about the agent's bias.  

\paragraph{Costly persuasion:}

We can forgo the singleton independence axiom in order to allow 
costly information acquisition.

\begin{thm} \label{costly}
A binary relation on menus $\succsim$ is a costly Bayesian persuasion
preference, represented by $(u,p_{0},c_{{\cal P}},\lambda)$, if and
only if it satisfies Axioms 1-9.
\end{thm}

Theorem 3 shows that when all axioms of Theorem 2 except the Singleton-menu Independence axiom (Axiom 10) are observed, this characteristic of the principal's preferences can be attributed to his anticipation of costly information acquisition (measured by $c_{\cal{P}}$), instead of constrained information acquisition (measured by $\Gamma$).

\paragraph{Sequential persuasion:}

We can remove the constant-menu independence axiom (Axiom 9) to allow for the sequential costly persuasion.\bigskip{}

\begin{thm} \label{sequential}
A binary relation on menus $\succsim$ is a sequential costly Bayesian persuasion
preference, represented by $(u,p_{0},c_{{\cal P}},c_{{\cal V}})$, if and
only if it satisfies Axioms 1-8.
\end{thm}

Theorem 4 shows that when all axioms of Theorem 3, except the Constant-menu Independence axiom (Axiom 9), are observed, then this feature of the principal's preferences can be attributed to his anticipation of not only acquiring costly information to persuade the agent, but also managing the agent's bias subject to costs (measured by $c_{\cal{V}}$). We note that the sequential model reduces to the self-regulation model of \cite{MO18} on the set of constant menus. Moreover, the sequential model can be directly related to the costly information acquisition model of \cite{dDMO17}; see section 5 for a detailed discussion.

We close this section with an elicitation of costs result and another result on the comparison of costs by using our characterization of the sequential costly Bayesian persuasion model given in Theorem 4 above.

\paragraph{Elicitation of costs:} The following result shows that we can elicit a pair of unique costs for the sequential model using choice data.

\begin{thm} \label{ident}
Let $\succsim$ be a sequential costly Bayesian persuasion preference represented by $(u,p_{0},c_{{\cal P}},c_{{\cal V}})$.
The function $c^{*}_{{\cal P}}:\Delta_0({\cal P})\to[0,\infty]$ defined for all $\tau\in\Delta_0({\cal P})$ by $c^{*}_{{\cal P}}(\tau)=\sup_{A\in\mathbb{A}}[u(x_A(\tau))-u(x_A)]$ and the function $c^{*}_{\cal{V}}:\Delta(\mathcal{V})\to[0,\infty]$ defined for all $\lambda\in\Delta(\mathcal{V})$ by $c^{*}_{\cal{V}}(\lambda)=\sup_{A\in\mathbb{A}^c}[b_{A}^{u}(\lambda)-u(x_A)]$ are unique \emph{minimal} cost functions such that $(u,p_{0},c^*_{{\cal P}},c^*_{{\cal V}})$ represents $\succsim$. Moreover, both cost functions are monotone, convex, lower semi-continuous, and grounded. 
\end{thm}

 We call the unique minimal cost functions $c^{*}_{{\cal P}}$ and  $c^{*}_{\cal{V}}$ identified above the canonical cost functions of the sequential model. To the best of our knowledge, our work is the first within the menu-choice literature that studies a decision making model with two layers of costs: a belief and a taste layer. A similar definition of a canonical cost function $c^*_{{\cal V}}$ over tastes appears in \cite{MO18}; they consider these costs for single-person decision making in the context of costly self-regulation unlike our model of two-person (principal and agent) decision making. The definition of the canonical cost function $c^*_{{\cal P}}$ over beliefs is more novel; although our definition is related to \cite{dDMO17}'s definition of a unique cost function $c^*$, unlike theirs, in our definition the first term within the supremum expression is not restricted to the constant equivalents of best acts (best in terms of the principal's commitment ranking), but rather uses the constant equivalents of induced menus. As such, our canonical cost function over beliefs stems from a richer choice behavior. Having said this, both \cite{dDMO17}'s and our canonical information costs satisfy the Blackwell monotonicity, convexity, lower semi-continuity, and groundedness. However, unlike \cite{dDMO17}, we do not claim that our canonical cost function over beliefs is the unique information cost function satisfying these properties. The main reason for this difference is that while \cite{dDMO17} consider a rich space with an unbounded utility to pin down the unique cost function with the aforementioned properties, we do consider only a bounded utility space.

\paragraph{Comparative statics:} The following result shows that, using menu-choice data, we can compare the canonical costs of two principals  who follow the sequential persuasion model. Let $\succsim_i$ and $\succsim_j$ be two sequential costly Bayesian persuasion preferences represented by $(u^i,p^i_{0},c^i_{{\cal P}},c^i_{{\cal V}})$ and $(u^j,p^j_{0},c^j_{{\cal P}},c^j_{{\cal V}})$, respectively. Suppose that the cost pairs $(c^i_{{\cal P}},c^i_{{\cal V}})$ and $(c^j_{{\cal P}},c^j_{{\cal V}})$ are canonical. 

\begin{thm} \label{comp}
   The following statements hold: (i) $\{x\}\succsim_i A \implies \{x\}\succsim_j A$ for all $A \in \mathbb{A}^c$ if and only if $u^i=u^j$ and $c^j_{{\cal V}}\geq c^i_{{\cal V}}$ ; and (ii) suppose $\succsim^c_i = \succsim^c_j$; then $\{x_A(\tau)\}\succsim_i A \implies \{x_A(\tau)\}\succsim_j A$ for all $A \in \mathbb{A}$ and $\tau\in \Delta(\cal{P})$ if and only if $p^i_{0} = p^j_{0}$ and $c^j_{{\cal P}}\geq c^i_{{\cal P}}$.
\end{thm}

Theorem 6 shows that by looking at menu-choice data, it is possible to compare the costs of different principals who follow the sequential model. In particular, the costs over tastes can be compared by focusing on the choice data using constant menus, whereas comparison of costs over beliefs requires not only the use of constant menus, but also the general menus which include non-constant acts.

\section{Discussion}

In this section, we discuss other implications of Bayesian persuasion for menu choice, as well as within-menu choice.

\subsection{Finite support uncertainty}

To understand the degree of uncertainty the principal has about the agent's bias, a definition similar to that in \citet[p.958]{DLR09} can be introduced. A non-empty closed set $A^{\ast }\subset \mathrm{co}(A)$ is called \textit{critical} for $A\in \mathbb{A}^{c}$ if $B\sim A^{\ast }$ for all $B$ with $A^{\ast }\subset \mathrm{co}(B)\subset \mathrm{co}(A)$. As in \cite{DLR09}, a critical subset of $A$ contains all lotteries that could be chosen from $A$. We first adopt the finiteness axiom in \cite{Stovall18}.\footnote{See \cite{Kopylov09}
for a related analysis.}

\begin{itemize}
\item \textbf{Axiom 11'.} [Stovall Finiteness] There exists $N\in \mathbb{N}$ such that for every $A$, 
there exists $A^{\ast }$ critical for $A$, where $\left\vert A^{\ast }\right\vert <N$.
\end{itemize}

Using Axiom 11', a special case of Bayesian persuasion with uncertain bias model can be characterized as follows:

\begin{prop} \label{Prop: N} 
A binary relation on menus $\succsim$ is a Bayesian persuasion with
uncertain bias preference, represented by $(u,p_{0},\Gamma,\lambda)$ 
with $\left\vert \mathrm{supp}(\lambda )\right\vert <N$,
if and only if it satisfies Axioms 1-10, and 11'. 
\end{prop}


While \cite{Stovall18}'s axiom provides a finite upper-bound for the cardinality of the principal's uncertainty, we can invoke \cite{DLR09}'s weaker finiteness axiom to have finite support, but without a certain upper-bound.

\begin{itemize}
\item \textbf{Axiom 11''.} [DLR Finiteness] For all $A$, there is a finite menu $A^{\ast }$ critical for $A$.
\end{itemize}

If we replace Axiom 11' with Axiom 11'', we obtain the following result:

\begin{prop}
A binary relation on menus $\succsim$ is a Bayesian persuasion with uncertain bias preference, represented by $(u,p_{0},\Gamma,\lambda)$ such that $\lambda$ has finite support, if and only if it satisfies Axioms 1-10, and 11''.
\end{prop}

\subsection{Costly information acquisition}

In applications of Bayesian persuasion, often the principal acquires costly information to persuade an agent with known bias $v$. This costly persuasion model can be axiomatized in our setting by adding Axiom 11 to the premise of Theorem \ref{costly}.

\begin{prop} 
A binary relation on menus $\succsim$ is a costly Bayesian persuasion (with known bias) preference, represented by $(u,p_{0},c_{{\cal P}},v)$, if and only if it satisfies Axioms 1-9 and 11. 
\end{prop} 

A further specification of this model can be given whenever there is no bias between the principal and the agent; that is, when $u=v$. An implication of this specific costly persuasion model can be given in the form of strategic rationality over binary constant menus (see \cite{Kreps79} for a related axiom).

\begin{itemize}
\item \textbf{Axiom $11'''$}. [Strategic Rationality] For all $a,b\in F^c$, $\{a\}\succsim \{b\}$
implies $\{a\}\sim \{a,b\}$.
\end{itemize}

Notice that the Strategic Rationality axiom (Axiom $11'''$) is stronger
than the Reducibility axiom (Axiom 11). In fact, if we replace Axiom
11 with Axiom $11'''$, we can characterize the costly information acquisition
model by appealing to our characterization results.
\begin{corr} \label{cia}
     A binary relation on menus $\succsim$ is a costly Bayesian persuasion (without bias) preference, represented by $(u,p_{0},c_{{\cal P}},u)$,
if and only if it satisfies Axioms 1-9 and $11'''$.
\end{corr} 

In the context of a single-decision maker with rational inattention, the above costly persuasion without bias preference was first axiomatized by \cite{dDMO17}. Corollary \ref{cia} provides an alternative characterization of their costly information acquisition model.

The two results above suggest that the costly Bayesian persuasion (with known bias) and costly Bayesian persuasion (without bias) preferences can be separated by the satisfaction or failure of a single axiom, Strategic Rationality. Specifically, given Axioms 1-9 and 11, if $\succsim$ satisfies $\{a\}\sim \{a,b\}$ for all $a,b \in F^c$ with $\{a\} \succ \{b\}$, then we have the costly information acquisition model; whereas, if $\{a\} \succ \{a,b\}$ for some $a,b \in F^c$ with $\{a\} \succ \{b\}$, then we have the costly persuasion model. In other words, while always satisfying a desire for flexibility (i.e. preference for a larger set) indicates single-person decision making (similar to \cite{dDMO17}'s axiomatization, where there is no bias), sometimes satisfying a desire for commitment (i.e. preference for a smaller set) indicates two-person decision making (as in Bayesian persuasion, where there is bias).

\subsection{Fixed information persuasion}

A particular case of the sequential persuasion model realizes when
the principal cannot flexibly acquire information, but rather must
use the same information structure regardless of the menu considered.
This type of sequential persuasion model will imply the following
exposure axiom.

\textbf{Axiom 8'.} [Neutral Exposure] There exists $\tau\in\Delta({\cal P})$
such that for all $A,B\in\mathbb{A}$, $\{x_{A}(\tau)\}\,1/2\,\{x_{B}\}\sim\{x_{A}\}\,1/2\,\{x_{B}(\tau)\}$.

Note that Axiom 8' is stronger than Axiom 8. Indeed, replacing Axiom
8 with Axiom 8' yields the above sequential persuasion model.
\begin{prop}
A binary relation on menus $\succsim$ is a sequential costly Bayesian
persuasion (with fixed information) preference, represented by $(u,p_{0},\tau,c_{{\cal V}})$ for some $\tau \in \Delta_0(\cal P)$,
if and only if it satisfies Axioms 1-7 and 8'.
\end{prop}
In a sequential persuasion with fixed information model, the principal
first acquires a fixed information structure regardless of the menu
considered, then conditional on the realized posterior, he tries to
manage the bias of the agent subject to costs. A more specific
case of the fixed information sequential persuasion model will be
when the principal cannot acquire any information at all. This specific
model will imply the following exposure axiom.

\textbf{Axiom 8''.} [Strong Neutral Exposure] There exists $p\in{\cal P}$
such that for all $A,B\in\mathbb{A}$, $\{x_{A}(\delta_{p})\}\,1/2\,\{x_{B}\}\sim\{x_{A}\}\,1/2\,\{x_{B}(\delta_{p})\}$.

Replacing Axiom 8 with Axiom 8'' yields the sequential persuasion
model without information acquisition.
\begin{corr} \label{noinfo}
A binary relation on menus $\succsim$ is a sequential costly Bayesian
persuasion (without information) preference, represented by $(u,p_{0},\delta_{p_{0}},c_{{\cal V}})$,
if and only if it satisfies Axioms 1-7, and 8''.
\end{corr}
We observe that the above no-information sequential persuasion model
(characterized in Corollary \ref{noinfo}) extends the self-regulation model
of \cite{MO18} over menus of lotteries to our setting of menus of acts.

\subsection{Value of information}

In Bayesian persuasion, the principal evaluates information instrumentally: sending an informative signal is valuable only when it induces the agent to take actions the principal prefers. With a taste wedge, a more informative signal does not automatically translate into more value, because revealing the state more accurately can move the agent's action away from what the principal wants. From this perspective, 
``value of information'' as measured by the principal's gain from committing to an informative signal relative to committing to no signal at all can be negative. In other words, if an informative signal would (on average) lead the agent to take an action that is worse for the principal than the agent's default action under the prior, the principal rationally chooses an uninformative signal. The following simple example illustrates this point.

Suppose that there are two states, $S=\{s_{1},s_{2}\}$ with prior belief $p_{0}=(1/2,1/2)$. There are three outcomes $X=\{x,y,z\}$. Let $u$ and $v$ represent the utility of the principal and agent on the outcomes, respectively. Suppose that $u(x)=1/\sqrt{2}$, $u(y)=-1/\sqrt{2}$, and $u(z)=0$; and $v(x)=0$, $v(y)=1/\sqrt{2}$, and $v(z)=-1/\sqrt{2}$. Thus, while the agent likes $y$ a lot, the principal dislikes $y$ and prefers $x$. Let $f=(x,x)$ and $g=(y,z)$ be two acts, and let $A=\{f,g\}$ be the menu available for the agent to pick an act from. Now, if the principal does not send a signal to the agent, then the agent chooses $f$ from A using her prior belief $p_{0}$; in this case, the principal obtains a payoff $U^{ni}(A)=1/\sqrt{2}$. But if the principal sends a perfectly informative signal, then the agent picks $g$ with probability $1/2$ that yields a payoff of $-1/\sqrt{2}$ to the principal, and picks $f$ with probability $1/2$ that yields a payoff of $1/\sqrt{2}$  to the principal; thus, in this case, the principal obtains an expected payoff $U^{fi}(A)=0$. Therefore, the value of information for the principal is negative $U^{fi}(A)-U^{ni}(A)=0-1/\sqrt{2}=-1/\sqrt{2}<0$. 

\subsection{Within menu choice behavior}

Bayesian persuasion models naturally yield choice data within menus. In the case of persuasion with fixed bias (known or uncertain to the principal), these choices will be deterministic conditional on the posterior; if the bias is endogenously realized (as in our sequential persuasion model), then they will be stochastic even when it is conditional on the posterior. In the Bayesian persuasion literature, only the models with known (and therefore fixed) bias have been considered so far. There are related works which axiomatize the Bayesian persuasion with known bias model by using (in part or in full) deterministic within menu choice data (see, e.g., \cite{Jakobsen21,Jakobsen24} or \cite{Mensch2025}). We now argue that some of the axioms considered for deterministic choice data can be violated if the principal can endogenously manage the agent's bias.

\paragraph*{Weak axiom of revealed preference:}

One of the well-known implications of Bayesian persuasion considered in the literature is the usual weak axiom of the revealed preference (WARP) condition. In our setting, WARP requires that if $f,g\in A\cap B$ and $f$ is chosen from $A$ and $g$ is chosen from $B$, then $f$ must be chosen from $B$. The self-regulation model of \cite{MO18} and our sequential model coincide over the set of constant menus $\mathbb{A}^{c}$. An obvious adaptation of Example 1 in \cite{MO18} can illustrate that the sequential model violates WARP. In particular, it is possible to have a choice situation where, for instance, only $f$ is chosen in $\{f,g,h\}$ and only $g$ is chosen in $\{f,g\}$ with $f,g,h\in F^{c}$. This is a violation of WARP once we observe that $A=\{f,g,h\}$ and $B=\{f,g\}$.

\paragraph*{Independence:}

Another prominent implication of the Bayesian persuasion considered in the literature is the Independence (IND) condition. In our setting, IND requires that if $f$ is chosen in $A\subset\mathbb{A}^{c}$, 
then $\alpha f+(1-\alpha) h$ will be chosen in $A\alpha\{h\}$ for any $\alpha\in(0,1)$ and $h\in F^{c}$. An obvious adaptation of Example 2 given in \cite{MO18} can show that IND is violated by the sequential model. In particular, it is possible to have the choice situation where, for instance, only $f$ is chosen in $\{f,g\}$, but only $\frac{1}{4}g+\frac{3}{4}h$ is chosen in $\{f,g\}\frac{1}{4}\{h\}$. This is a violation of the IND once we observe that $A=\{f,g\}$ and $\alpha=\frac{1}{4}$.

\section{Conclusion}

This paper has developed an axiomatic foundation for Bayesian persuasion based on observable menu-choice data by the principal. 
By introducing and analyzing monotonicity axioms that separately discipline the roles of tastes and information, we have shown that a broad class of Bayesian persuasion models share a common behavioral core. Within this framework, the principal's possible uncertainty about the agent's bias, and his ability to manage that bias after learning payoff-relevant information about the state, are explicitly incorporated. Our characterization theorems demonstrate that the models we study all satisfy these axioms and thus can be understood as different instantiations of a unified menu-preference approach to Bayesian persuasion. In this sense, we move beyond the dichotomy between ``negative'' and ``positive'' states and explain how a single decision-maker may rationally demand both commitment and flexibility, depending on the interaction between uncertainty about the state and concern about the agent's bias.

Building on this axiomatic structure, we propose a constructive elicitation method that recovers the principal's subjective costs of acquiring information and managing bias solely from menu-choice data. In the sequential costly Bayesian persuasion model, these costs are identified through the exposure axiom, which ensures that for every menu there exists an information structure under which the relative value of that information is maximized with that menu. Our method extends existing approaches by evaluating the full induced menu rather than just the best act at each posterior, thereby accommodating strategic interaction between a principal and an agent with potentially misaligned tastes. This delivers an elicitation procedure that is robust to endogenous bias between principal and agent, and nests as special cases earlier identification results on information and self-regulation costs. Finally, we situate our analysis alongside recent work that axiomatizes Bayesian persuasion using stochastic choice and richer observables, viewing our contribution as complementary: we offer a new sequential persuasion model and a menu-based elicitation technique that can, in future work, be mapped into these stochastic choice frameworks to derive testable ex-post implications and guide empirical applications of persuasion with costly information acquisition and bias management.

\appendix

\section{Appendix}

In this section, we first provide the axiomatic characterization of
a general model of delegation, where the principal not only controls
the posterior belief, but also manages the uncertain bias of the agent. We then provide proofs of the results stated in the body of the paper.

\subsection{\label{subsec:Preliminaries}Preliminaries}

Let $\Sigma$ denote the Borel sigma-algebra over ${\cal P \times \cal V}$, and
let $B(\Sigma)$ be the set of bounded $\Sigma$-measurable functions
mapping $\cal P \times\mathcal{V}$ to $\mathbb{R}$. When endowed with the sup-norm
metric, $B(\Sigma)$ is a Banach space. The topological dual of $B(\Sigma)$
is the space $ba(\Sigma)$ of all bounded and finitely-additive set
functions $\mu:\Sigma\to\mathbb{R}$, the duality being $\left\langle \varphi,\mu\right\rangle =\int_{\cal P \times \mathcal{V}}\varphi(p,v)\,\mu(dp\times dv)$
for all $\varphi\in B(\Sigma)$ and all $\mu\in ba(\Sigma)$ (see,
e.g., \citet[p. 258]{DS58}). For $\varphi,\psi\in B(\Sigma)$,
we write $\varphi\ge\psi$ if $\varphi(p,v)\ge\psi(p,v)$ for all $p \in \cal P$ and $v\in\mathcal{V}$. 

Let $\Phi$ be a non-empty subset of $B(\Sigma)$, and $\Phi_{c}$
be the constant functions in $\Phi$. Set $\Phi$ is called a \emph{tube}
if $\Phi=\Phi+\mathbb{R}$. A functional $I:\Phi\to\mathbb{R}$ is\emph{
(i) normalized} if $I(k)=k$ for all $k\in\Phi_{c}$, \emph{ (ii) monotone} if $\varphi\ge\psi$ implies $I(\varphi)\ge I(\psi)$
for all $\varphi,\psi\in\Phi$,\emph{ (iii) translation invariant}
if $I(\alpha\varphi+(1-\alpha)k)=I(\alpha\varphi)+(1-\alpha)k$ for
all $\varphi\in\Phi$, $k\in\Phi_{c}$, and $\alpha\in[0,1]$, such
that $\alpha\varphi,\alpha\varphi+(1-\alpha)k\in\Phi$,\emph{ (iv)
vertically invariant} if $I(\varphi+k)=I(\varphi)+k$ for all $\varphi\in\Phi$
and $k\in\Phi_{c}$ such that $\varphi+k\in\Phi$, and a \emph{(v)}
\emph{niveloid} if $I(\varphi)-I(\psi)\le\underset{(p,v)\in \cal P \times \mathcal{V}}{\sup}(\varphi(p,v)-\psi(p,v))$ for all $\varphi,\psi\in\Phi$. Clearly, a niveloid is Lipschitz continuous. Moreover, \citet{CMMR14} show that a niveloid is monotone vertically invariant functional, while the converse is true whenever its domain is a tube. Denote by $\bar{a}=(1/n,...,1/n)$ the uniform distribution over $X$. Let $\mathcal{P}^\circ$ denote the interior of $\cal P$ (i.e., the set of lotteries with full support) and $F^\circ$ be the set of acts whose range is $\mathcal{P}^\circ$. Let $\mathbb{A}^{o}\subset\mathbb{A}$ be the collection of closed non-empty subsets of $F^\circ$.

For $u\in{\cal V}$ and $A\in\mathbb{A}$, define $\phi_{A}^{u}:\cal P \times \mathcal{V}\to\mathbb{R}$
by $\phi_{A}^{u}(p,v)=\max_{f\in m_{v,p}(A)}u(f).p$ for all $p \in \cal P$ and $v\in\mathcal{V}.$
When $u$ is clear from the context, we omit the superscript $u$.
By the Maximum Theorem (see, e.g., \citet[pp. 569--570]{AB06}),
$\phi_{A}$ is an upper semi-continuous function taking values in $K=[u_{*},u^{*}]$,
where $u_{*}=\underset{x\in X}{\min}\thinspace u(x)$ and $u^{*}=\underset{x\in X}{\max}\thinspace u(x)$.
The upper semi-continuous functions are $\Sigma$-measurable (\citet[pp. 184--186]{Bill95}).
As a result, $\phi_{A}\in B(\Sigma,K)$, where $B(\Sigma,K)$ denotes
the functions in $B(\Sigma)$ assuming values in $K$. Let $\Phi=\{\phi_{A}\,:\,A\in\mathbb{A}\}$
and $\Phi^{o}=\{\phi_{A}\,:\,A\in\mathbb{A}^{o}\}$. Clearly $0\in\Phi^{o}$
and $\Phi^{o}\subseteq\Phi$. Moreover, since $\phi_{A\alpha B}=\alpha\phi_{A}+(1-\alpha)\phi_{B}$
for any $A,B\in\mathbb{A}$ and $\alpha\in[0,1]$, both $\Phi$ and
$\Phi^{o}$ are convex sets. It is straightforward to show that $\phi_{A}=\phi_{co(A)}$
for all $A\in \mathbb{A}$. 

\subsection{Costly delegation model}

We take the perspective that delegation is a way of affecting
the agent's information, as well as bias by inducing a random realization
of beliefs and tastes for the agent. This reflects the idea that while
the principal's tastes are stable, delegation can give rise to compositions
of random tastes and beliefs for the agent.

Formally, let $\Delta({\cal P}\times\mathcal{V})$ be the set of all
probability distributions on ${\cal P}\times\mathcal{V}$, with typical elements $\pi,\rho,\mu\in\Delta({\cal P}\times\mathcal{V})$
called \emph{distributions} (over posteriors and utilities). We endow $\Delta({\cal P}\times\mathcal{V})$ with the weak*-topology. The principal
would like to use a delegation strategy, so as to have an agent who
can make informed and aligned choices. Given a menu $A$, let $b_{A}^{u}(\pi)$
denote the \emph{benefit of delegation} the principal (with utility
$u$) anticipates when his delegation strategy induces a distribution
$\pi$ over beliefs and tastes.

Formally, let $b_{A}^{u}:\Delta({\cal P}\times{\cal V})\to\mathbb{R}$
be the function s.t. for any $\pi\in\Delta({\cal P}\times{\cal V})$,
\[
b_{A}^{u}(\pi)=\int_{{\cal {\cal P}\times V}}\,[\,\max_{f\in m_{v,p}(A)}u(f)\cdot p\,]\,\pi(dp\times dv).
\]

The principal must incur the \emph{costs of delegation }in order to
exploit the benefits of imposing the agent's beliefs and tastes. As
such, the principal needs to balance the benefits and costs of delegation
when deciding how much effort to exert. These costs are represented
by a function $c:\Delta({\cal P}\times{\cal V})\to[0,\infty]$, where
$c(\pi)$ is a behavioral measure of the effort required to induce
a delegation strategy $\pi$ such that $c(\delta_{p_{0}}\times\hat{\lambda})=0$
for some $p_{0}\in{\cal P}$ and $\hat{\lambda}\in\Delta({\cal V})$
satisfying $c(\delta_{p_{0}}\times\hat{\lambda})\leq c(\pi)$
for all $\pi\in \Delta({\cal P}\times\mathcal{V})$.

\paragraph*{Delegation representation:}

Given $u\in\mathcal{V}$ and $c:\Delta({\cal P}\times\mathcal{V})\to[0,\infty]$, as well as $p_{0}\in{\cal P}$ and $\hat{\lambda}\in\Delta({\cal V})$,  the principal chooses an optimal delegation strategy $\pi^{*}$ by weighting the benefits and costs of distribution to evaluate a menu $A\in\mathbb{A}$. As such, we can define a function
$U:\mathbb{A}\to\mathbb{R}$ on the set of menus that satisfies for all
$A\in\mathbb{A}$, 
\[
U(A)=\max_{\pi\in\Delta({\cal P}\times\mathcal{V})}\left[b_{A}^{u}(\pi)-c(\pi)\right]
\]
which provides a representation for the principal's preferences $\succsim$
over $\mathbb{A}$. In that case, we say that the principal has a\emph{ costly delegation
preference} represented by $(u,p_{0},c,\hat{\lambda})$.

\paragraph*{Dominance:}

Let $A\trianglerighteqslant B$ ($A$ dominates $B$) if  $\mathrm{co}(A)^p_{a}\,1/2\,\mathrm{co}(B)^p_{b}\,\supset\,\mathrm{co}(B)^p_{a}\,1/2\,\mathrm{co}(A)^p_{b}$ for all $p \in \cal{P}$ and for all $a,b\in \Delta(X)$ with $\{a\}\succ\{b\}$, where $C_{d}=\left\{ c\in C:\{c\}\sim\{d\}\right\} $ is the set of lotteries in constant menu $C\in\mathbb{A}^c$ which are indifferent to the lottery $d\in \Delta(X)$. Under Axiom 4, our taste and information dominance axioms, Axiom 5 and Axiom 6, together imply the following dominance axiom.

\begin{itemize}
\item \textbf{Axiom 5'.} [Dominance] If $A\trianglerighteqslant B$, then $A\succsim B$.\bigskip{}
\end{itemize}

Although,  under Axiom 4, Axiom 5' implies Axiom 5, it is not necessarily true that Axiom 5' implies Axiom 6. The following result characterizes the costly delegation preference by invoking the weaker Axiom 5', instead of using Axiom 5 and Axiom 6 together.\bigskip{}

\begin{thm} \label{del}
A binary relation on the set of menus $\succsim$ is a costly delegation preference,
represented by $(u,p_{0},c,\hat{\lambda})$, if and only if it satisfies
Axioms 1-4, 5', and 7.
\end{thm}

\begin{proof}
\emph{[Necessity]:} It is clear that Axioms 1-4, and 7 follow from the costly delegation model. We therefore omit them. Let $A,B\in\mathbb{A}$ be such that $A\trianglerighteqslant B$. By definition, $\mathrm{co}(A)^p\trianglerighteqslant_{\cal{V}} \mathrm{co}(B)^p$ for all $p\in \mathcal{P}$. By Lemma 1 in \cite{MO18}, we have $b^u_{\mathrm{co}(A)^p}\geq b^u_{\mathrm{co}(B)^p}$ for all $p\in \mathcal{P}$. Since 
    \begin{align*}
        \int_{V}[\max_{f^{p}\in\arg\max_{g^{p}\in \mathrm{co}(A)^{p}}v(g^{p})}u(f^{p})]\rho (dv)
        =\int_{V}[\max_{f\in m_{v,p}(A)}\sum u(f(s))p(s)]\rho (dv)
    \end{align*}
for any $\rho\in\Delta(\mathcal{V})$, we have $b^u_A(\pi(\{p\}\times \mathcal{V}))\geq b^u_A(\pi(\{p\}\times \mathcal{V}))$ for all $p\in \mathcal{P}$ and $\pi\in\Delta(\mathcal{P}\times \mathcal{V})$.
This implies that $b^u_{A}(\pi) = \int_{\cal{P}}b^u_{A}(\pi(\{p\}\times \mathcal{V}))dp \geq \int_{\cal{P}}b^u_{B}(\pi(\{p\}\times \mathcal{V}))dp = b^u_{B}(\pi)$ for all $\pi\in\Delta(\mathcal{P}\times \mathcal{V})$. Hence, we obtain $b^u_A\geq b^u_B$, and so $A\succsim B$ by the representation.

\emph{[Sufficiency]:} We first show that Axioms 1-3, 5', and 7 imply that the commitment ranking can be represented by a utility $u\in{\cal V}$ and prior $p_{0}\in{\cal P}$ such that $\{f\}\succsim\{g\} $ if and only if $u(f)\cdot p_{0}\ge u(g)\cdot p_{0}$ for all $f,g\in F$. Let $f,g\in F$ and assume that $\{f\}\sim\{g\}$. By Axiom 7, we have $\{f\}\succsim \{g\}\, 1/2\,\{f\}$, and so, by Axiom 3, $\{g\}\, 1/2\,\{f\}\succsim \{g\}$. Thus, we must have $\{f\}\sim\{g\}\, 1/2\,\{f\}$ implying, by Axiom 3, that $\{f\}\, 1/2\,\{h\}\sim \{g\}\, 1/2\,\{h\}$ for all $h\in F$. Thus, by \cite{HM53}, there exists $u\in\cal V$ such that $\{a\}\succsim\{b\}$  if and only if $u(a)\ge u(b)$ for all $a,b\in F^c$.
Now let $f,g\in F$ such that $\{f(s)\}\succsim \{g(s)\}$ for all $s\in S$, and so
$u(f(s))\geq u(g(s))$ for all $s\in S$. As such, we have $\sum_{s\in S}u(f(s))p(s)\geq \sum_{s\in S}u(g(s))p(s)$ for all $p\in \Delta(S)$.
Since $u$ is affine, we must have $u(f^p)\geq u(g^p)$ for all $p\in \Delta(S)$, and so $\{f^p\}\succsim\{g^p\}$ for all $p\in \Delta(S)$. Thus, clearly we have $\{f^p\}\trianglerighteqslant_{\cal V}\{g^p\}$ for all $p\in \Delta(S)$ which implies that $\{f\}\trianglerighteqslant\{g\}$.
By Axiom 5', we have $\{f\}\succsim\{g\}$, and so by \cite{AA63}, there exists a prior $p_0\in\mathcal{P}$ such that $\{f\}\succsim\{g\}$  if and only if $u(f)\cdot p_{0}\ge u(g)\cdot p_{0}$ for all $f,g\in F$. 

We now show that every $A\in \mathbb{A}$ has a singleton equivalent $x_A\in \Delta(X)$ such that $\{x_A\}\sim A$. Since $S$ is finite, we can take the best lottery $\overline{x}_f$ and the worst lottery $\underline{x}_f$ that may occur in $f$. For any $A\in \mathbb{A}$, define $B_A\equiv \{\overline{x}_f\in \Delta(X):f\in A\}$ and $W_A\equiv \{\underline{x}_f\in \Delta(X):f\in A\}$. Note that $B_A$ and $W_A$ are compact because they are closed subsets of $\Delta(X)$. Since $u$ is continuous on $\Delta(X)$, there exist a best lottery $\overline{x}_A \in B_A$ and a worst lottery $\underline{x}_A \in W_A$ that may occur in $A$. By definition,  we have $\{\overline{x}_{A}^{p}\} \trianglerighteqslant_{\cal{V}} A^p \trianglerighteqslant_{\cal{V}} \{\underline{x}_{A}^{p}\}$ for any $p\in \Delta(X)$. Axiom 5' implies $\{\overline{x}_A\}\succsim A\succsim \{\underline{x}_A\}$. Hence, as in Claim 2 in \cite{dDMO17}, Axiom 2 implies that there exists $\alpha\in [0,1]$ such that $\{\overline{x}_A\}\alpha \{\underline{x}_A\}\sim A$. Let $x_A = \overline{x}_A \alpha \underline{x}_A$.

To establish the desired representation, we show that there is a normalized convex niveloid $I:\Phi\to\mathbb{R}$ such that, for all menus $A$ and $B$, $A\succsim B$ if and only if $I(b^u_{A})\ge I(b^u_{B})$. Following the approach in \citet{MMR06}, an application of Fenchel-Moreau duality then establishes $I(b^u_{A})=\max_{\pi\in\Delta(\mathcal{P} \times \mathcal{V})}\left(\left\langle b^u_{A},\pi\right\rangle -c^{*}(\pi)\right)$ for all $A\in\mathcal{A}$. For technical reasons, we start by defining a functional $I^{o}$ on $\Phi^{o}$, and then use Axiom 2 to extend the functional to $\Phi$. Let $I^\circ:\Phi^\circ \rightarrow \mathbb{R}$ be a functional defined by $I^{\circ}(b_{A}^{u})=u(x_A)$ for all $A\in\mathbb{A}^\circ$. For any $A,B\in\mathbb{A}^\circ$ with singleton equivalents $x_A$ and $x_B$,
$A\succsim B$ if and only if $\{x_A\}\succsim \{x_B\}$. Hence, $I^\circ(b_{A}^{u})\geq I^\circ(b_{B}^{u})$ if and only if $A\succsim B$. We need to show that $I^\circ$ is well-defined; that is, $I^\circ(b^u_A)=I^\circ(b^u_B)$ for any $A,B\in\mathbb{A}^\circ$ with $b^u_A=b^u_B$. The key step is to show that $I$ is monotone; that is, $b_{A}^{u}\ge b_{B}^{u}$ implies $I(b_{A}^{u})\ge I(b_{B}^{u})$. Let $A,B\in \mathbb{A}^\circ$ be such that $b^u_A\geq b^u_B$. Since $b^u_A=b^u_{\mathrm{co}(A)}$ and $b^u_B=b^u_{\mathrm{co}(B)}$, we have $b^u_{\mathrm{co}(A)}\geq b^u_{\mathrm{co}(B)}$. This implies that for $p\in \mathcal{P}$ and $\pi\in\Delta(\mathcal{P}\times \mathcal{V})$, $b^u_{\mathrm{co}(A)}(\pi(\{p\}\times \mathcal{V}))\geq b^u_{\mathrm{co}(B)}(\pi(\{p\}\times \mathcal{V}))$, and so $b^u_{{\mathrm{co}(A)}^p}\geq b^u_{{\mathrm{co}(B)}^p}$ for all $p\in\mathcal{P}$. As such, by the necessity part of Lemma 1 in \cite{MO18}, we have $\mathrm{co}(A)^p\trianglerighteqslant_{\cal{V}} \mathrm{co}(B)^p$ for all $p\in\mathcal{P}$. By Axioms 4 and 5', we have $A \succsim B$, and so $U(A)\geq U(B)$ implying that $I^\circ$ is monotone. Moreover, the obvious adaption of arguments in \cite[pp.39-41]{MO18} show that $I^\circ$ is a normalized convex niveloid.

We now extend $I^{o}$ to $\Phi$. For any menu $A\in\mathbb{A}$ and number $m\in\mathbb{N}$, define $A^{m}=A\frac{m-1}{m}\{\bar{a}\}$ where $\bar{a}=(1/n,...,1/n)$ is the uniform distribution over $X$. Note that for all $A\in\mathbb{A}$ and $m\in\mathbb{N}$, $A^{m}\in\mathbb{A}^{o}$ and $b^u_{{A}^{m}}\to b^u_{A}$ uniformly as $m\to\infty$. Define the a functional $I:\Phi\to\mathbb{R}$ by $I(b^u_{A})=\lim_{m\to\infty}I^{o}(b^u_{{A}^{m}})$ for all $A\in\mathcal{A}$. Since $I^{o}$ is a niveloid, it is a continuous function, and so $I^{o}$ preserves convergence. Thus, for any menu $A\in\mathbb{A}$, the sequence $\{I^{o}(b^u_{{A}^{m}})\}_{m\in\mathbb{N}}$ converges
to a point in $[u_{*},u^{*}]$ showing that $I$ is well-defined. Moreover, following the arguments in \cite[pp.41-42]{MO18} we can immediately establish that (i) $I$ preserves the properties of $I^{o}$, i.e., it is also a normalized convex niveloid which assumes values in $K=[u_{*},u^{*}]$ and (ii) $I$ satisfies the property that for all $A,B\in\mathbb{A}$, $A\succsim B$ if and only if $I(b^u_{A})\ge I(b^u_{B})$. Since $\Phi$ is a convex subset of $B(\Sigma,K)$ and $I$ is a normalized convex niveloid, the obvious adaption of the arguments in the proof of \citet[Lemma 27]{MMR04} then establishes that $I(b^u_A)=\max_{\pi\in\Delta(\mathcal{P} \times \mathcal{V})}\left(\left\langle b^u_A,\pi\right\rangle -c^{*}(\pi)\right)$ for all $\varphi\in\Phi$, where $c^*$ is a non-negative lower semi-continuous and proper cost function such that $c^{*}(\pi)=\sup_{A\in\mathbb{A}}\left(\left\langle b^u_A,\pi\right\rangle -u(x_{A})\right)$ for all $\pi\in\Delta(\cal P \times \mathcal{V})$. Moreover, $c^*$ can be normalized such that  $c^*(\delta_{p_{0}}\times\hat{\lambda})=0$ for the prior belief $p_{0}\in{\cal P}$ and some distribution $\hat{\lambda}\in\Delta({\cal V})$ over utilities.
\end{proof}

\subparagraph*{Identification:}
It is a standard argument to show that for a given costly delegation preference $\succsim$ represented by $(u,p_{0},c,\hat{\lambda})$, utility $u$ is the unique function in $\mathcal{V}$ and $p_{0}$ is the unique prior belief in ${\cal P}$ representing the singleton ranking. However, in general cost function $c$ is not unique; that is, there can be more than one cost function associated with the same preference
relation $\succsim$. But, one can obtain a unique \emph{minimal} cost function $c^*$ (i.e., $c^*(\pi)\leq c(\pi)$ for all $c$ and $\pi$) by using data on singleton equivalent menus.

Let $\succsim$ be a costly delegation preference represented by $(u,p_{0},c,\hat{\lambda})$.

\begin{thm}
The function $c^{*}:\Delta({\cal P}\times\mathcal{V})\to[0,\infty]$,
defined by
\[
c^{*}(\pi)=\sup_{A\in\mathbb{A}}\left(b_{A}^{u}(\pi)-u(f_{A})\right)\quad\forall\pi\in\Delta({\cal P}\times\mathcal{V}),
\]
is the unique \emph{minimal} cost function such that $(u,p_{0},c^{*},\hat{\lambda})$
represents $\succsim$.
\end{thm}

\begin{proof}
The proof of this result follows from arguments given in the proof of the identification result, Theorem 2, in \cite{MO18}. Therefore, we omit them here.
\end{proof}

We call the above unique minimal cost function $c^{*}$, the \emph{canonical cost} and we call $(u,p_{0},c^{*},\hat{\lambda})$ canonical representation of the costly delegation model. Clearly, it is true from the above formula that $c^{*}(\pi)\geq0$ for all $\pi$; $c^{*}$ is \emph{convex} and \emph{lower semi-continuous} given that it is the supremum of continuous affine functions. Moreover, $c^*$ can be normalized such that $c^*(\delta_{p_{0}}\times\hat{\lambda})=0$
for some $\hat{\lambda}\in\Delta({\cal V})$ such that $c^*(\delta_{p_{0}}\times\hat{\lambda})\leq c^*(\pi)$
for all $\pi\in\Delta({\cal P}\times {\cal V})$.

\subparagraph*{Comparative statics:}

The following defines a comparative behavior that formalizes when
a principal $PR_{j}$ finds commitment more valuable (i.e., flexibility
less valuable) than another principal $PR_{i}$ in terms of menu-choice
data.\medskip{}

\textbf{Definition:} We say $PR_{j}$ has a \emph{stronger desire for commitment}
than $PR_{i}$ if it is true that $\{f\}\succsim_{i}A\;\Rightarrow\;\{f\}\succsim_{j}A$ for all $A\in \mathbb{A}$.

This comparative characterizes when one principal has always higher
costs of delegation than the other. Let $PR_{i}$ and $PR_{j}$ have costly delegation preferences with canonical representations $(u_{i},p_{i},c_{i}^{*},\hat{\lambda}_{i})$
and $(u_{j},p_{j},c_{j}^{*},\hat{\lambda}_{j})$, respectively.

\begin{thm} \label{comp-del}
$PR_{j}$ has a \emph{stronger desire for commitment}
than $PR_{i}$ if and only if $u_{j}=u_{i},\;p_{j}=p_{i}\;,\hat{\lambda}_{j}=\hat{\lambda}_{i},\text{and}\;c_{j}^{*}\ge c_{i}^{*}$.
\end{thm}

\begin{proof}
The proof of this result follows from similar arguments used in the proof of the comparative statics result, Theorem 3, given in \cite{MO18}. Therefore, we omit them here.
\end{proof}

\subsection{Proofs of the results in the text}

We now provide proofs of the results stated in the body of the paper.

\paragraph*{Proof of Theorem \ref{known}:}

The persuasion with known bias preference $\succsim$ represented by the parameters $(u,p_{0},\Gamma,v)$  is a special case of the persuasion with uncertain bias preference represented by the parameters $(u,p_{0},\Gamma,\lambda)$. As such, both preferences satisfy Axioms 1-10. We now argue that Axiom 11 pins down the persuasion with known bias preference within this class. Clearly, when  $\succsim$  is represented by the parameters $(u,p_{0},\Gamma,v)$, then its restriction to the set of constant menus satisfies the Reducibility axiom, Axiom 11.

Now suppose that $\succsim$ satisfies Axiom 11. As such, for all $a,b\in F^{c}$, we have $\{a,b\}\sim\{a\}$ or $\{a,b\}\sim\{b\}$. Suppose for contradiction that there exists $v_{1},v_{2}\in \mathrm{supp}(\lambda)$ such that $v_{1}\neq v_{2}$. Let $\theta=v_{1}-v_{2}$. We have $v_{1}.\theta=1-v_{1}v_{2}>0$ and $v_{2}.\theta=v_{2}.v_{1}-1<0$ since $v_{1}.v_{2}<1$ by the Cauchy-Schwarz inequality. By continuity, there exists open neighborhoods, $N_{1}$ around $v_{1}$ and $N_{2}$ around $v_{2}$ such that $v.\theta>0$ for all $v\in N_{1}$ and $v.\theta<0$ for all $v\in N_{2}$. Since $v_{1},v_{2}$ are in the support, we have $\lambda(N_{1})>0$ and $\lambda(N_{2})>0$.

Pick a lottery $b$ in the interior of $\Delta(X)$; that is,
$b(x)>0$ for all $x\in X$. Let $\epsilon>0$ be small enough such
that $a=b+\epsilon\theta\in\Delta(X)$. Then for $v\in N_{1}$, we
have $v(a)>v(b)$, while for $v\in N_{2}$ we have $v(b)>v(a)$. Moreover,
without loss of generality, we can take $u(a)\neq u(b)$. Otherwise,
if $u.\theta=0$, then we can let $\theta'=\theta+r\,u$ for some
sufficiently small $r\neq0$, and so we will have $u(a)-u(b)=\epsilon[u.\theta+r]=\epsilon r\neq0$. In this case, we also have $v_{1}.\theta'=1-v_{1}v_{2}+r\,v_{1}.u>0$
and $v_{2}.\theta'=v_{2}.v_{1}-1+r\,v_{2}.u<0$ provided that $r$
is sufficiently small.

Therefore, we have both $\lambda_{a}>0$ and $\lambda_{b}>0$ where
$\lambda_{a}=\lambda(v\in{\cal V}:v(a)>v(b))$ and $\lambda_{b}=\lambda(v\in{\cal V}:v(b)>v(a))$.
But then, $U(\{a\})>U(\{a,b\})$ and $U(\{a,b\})>U(\{b\})$ since
$U(\{a,b\})=\lambda_{a}u(a)+\lambda_{b}u(b)+\lambda_{ab}\max\{u(a),u(b)\}$ where $\lambda_{ab}=\lambda(v\in{\cal V}:v(a)=v(b))$.
As such, $\{a\}\succ\{a,b\}\succ\{b\}$ which violates Axiom 11. Thus,
there must be a unique utility $v\in \mathrm{supp}(\lambda)$ showing that
$\succsim$ is a Bayesian persuasion with known bias preference. \qed

\paragraph*{Proof of Theorem \ref{uncertain}:}

The persuasion with uncertain bias preference represented with parameters $(u,p_{0},\Gamma,\lambda)$  is a special case of the costly persuasion preference represented with parameters $(u,p_{0},c_{{\cal P}},\lambda)$. In particular, we have $c_{\cal P}(\tau)=0$ if $\tau \in \Gamma$ and $c_{\cal P}(\tau)=\infty$ if $\tau \notin \Gamma$. We know by Theorem \ref{costly} that Axioms 1-9 are necessary and sufficient for the costly persuasion model. It is clear that the persuasion with uncertain bias model satisfies Axiom 10. But then the converse direction follows from the proof of Corollary 1 given in \cite{dDMO17} after observing that Axiom 10 directly implies Axiom $10^*$: $A \sim A\alpha \{x_A\}$ for all $A\in \mathbb{A}$.\qed

\paragraph*{Proof of Theorem \ref{costly}:}

The costly Bayesian persuasion model  $(u,p_{0},c_{{\cal P}},\lambda)$ is a special case of the sequential costly Bayesian persuasion model  $(u,p_{0},c_{{\cal P}},c_{{\cal V}})$. We know by Theorem \ref{sequential} that Axioms 1-8 are necessary and sufficient for the sequential model. We now show that Axiom 9 is necessary and sufficient to pin down the costly Bayesian persuasion model. We know, as noted in the proof of Theorem \ref{del}, that for each $p\in\Delta(S)$, the binary relation $\succsim_{p}$ restricted to the space of constant menus $\mathbb{A}^c$ has a representation  $V_{p}$ such that $V_{p}(A)=V(A^{p})=\max_{\lambda\in\Delta(\mathcal{V})}[b_{A^{p}}^{u}(\lambda)-c_{{\cal V}}(\lambda)]$ for some $c_{{\cal V}}$. But then by Proposition 1 in \cite{MO18}, the relation $\succsim_{p}$ satisfies Axiom 9 if and only if it has a random Strotz representation such that $V_{p}(A)=b_{A^{p}}^{u}(\lambda)$ for some $\lambda\in\Delta({\cal V})$. \qed

\paragraph*{Proof of Theorem \ref{sequential}:}

\emph{[Necessity]:} We start our proof by showing the necessity of the axioms. Since $\succsim$ is a costly delegation preference, by Theorem \ref{del}, we already know that Axioms 1-4, 7, and 5' (and so 5) are necessary. To see that Axiom 6 holds, let $A,B\in\mathbb{A}$ such that $A\trianglerighteqslant_{\cal{P}} B$. This implies that for any $\tau \in \Delta(\cal{P})$, we have $u(x_{A}(\tau))=\int u(x_{A^{p}})\,\tau(dp) \geq \int u(x_{B^{p}})\,\tau(dp) = u(x_{B}(\tau))$ which implies that
\[
\max_{\tau\in\Delta_0({\cal P})}\left[\int_{{\cal P}}u(x_{A^p})\,\tau(dp)\,-\,c_{{\cal P}}(\tau)\right] \geq \max_{\tau\in\Delta_0({\cal P})}\left[\int_{{\cal P}}u(x_{B^p})\,\tau(dp)\,-\,c_{{\cal P}}(\tau)\right],
\]
and so we have $A \succsim B$.

Now, let $A \in \mathbb{A}$. By the sequential model, $U(A)=\sup_{\tau \in \Delta_0(\cal{P})}[U(x_{A}(\tau))-c_{{\cal P}}(\tau)]$,  and so  $U(x_{A}(\tau^{*}))-U(A)\geq U(x_{B}(\tau^{*}))-U(B)$ for all $B$ since the supremum is achieved at some $\tau^{*}\in\Delta_{0}({\cal P})$ implying that $\frac{1}{2}U(\{x_A(\tau^*)\})+\frac{1}{2}U(\{x_B\})\geq \frac{1}{2}U(\{x_A\})+\frac{1}{2}U(\{x_B(\tau^*)\})$. Since $U$  is linear over singleton menus, we have $U(\frac{1}{2}\{x_A(\tau^*)\}+\frac{1}{2}\{x_B\})\geq U(\frac{1}{2}\{x_A\}+\frac{1}{2}\{x_B(\tau^*)\})$. Since $U$ represents $\succsim$, we have $\frac{1}{2}\{x_A(\tau^*)\}+\frac{1}{2}\{x_B\}\succsim
\frac{1}{2}\{x_A\}+\frac{1}{2}\{x_B(\tau^*)\}$ for all $B$. Now let $C=\{f\}$ and $D=\{g\}$ be two singleton menus. Since $\tau^* \in \Delta_0(\cal P)$ and $U$ is linear over singleton menus, we have $U(x_C(\tau^*))=u(x_C(\delta_{p_0}))$ and $U(x_D(\tau^*))=u(x_D(\delta_{p_0}))$. By the sequential model, we have $U(x_C)=U(C)=u(f).p_0-c(\delta_{p_0})=u(f).p_0=U(x_C(\delta_{p_0}))$ and similarly, $U(x_D)=U(D)=u(g).p_0-c(\delta_{p_0})=u(g).p_0=U(x_D(\delta_{p_0}))$. Since $U$ is linear over singleton menus, we have $U(\{x_C\}1/2\{x_D(\delta_{p_0}\}))=U(\{x_D\}1/2\{x_C(\delta_{p_0}\}))$ showing that Axiom 8 holds.

\emph{[Sufficiency]:} We now show the sufficiency of the axioms. Recall that Axioms 5 and 6 together with Axiom 4 imply Axiom 5'. As such, by Theorem \ref{del}, using axioms 1-7, we know that $\succsim$ is a costly delegation preference represented by some $U$ with parameters $(u,p_{0},c,\hat{\lambda})$. For each $p\in\Delta(S)$, define a binary relation $\succsim_{p}$ over $\mathbb{A}$ such that $A\succsim_{p}B$ if $A^{p}\succsim B^{p}$. The binary relation $\succsim_{p}$ restricted to the space $\mathbb{A}^{c}$ satisfies Axioms 1-6 as stated in \cite{MO18}'s framework, and so $\succsim_{p}$ has a representation over $\mathbb{A}$ such that $V_{p}(A)=V(A^{p})=\max_{\lambda\in\Delta(\mathcal{V})}[b_{A^{p}}^{u}(\lambda)-c_{{\cal V}}(\lambda)]$ for some $c_{{\cal V}}$. Since any two preferences $\succsim_{p}$ and $\succsim_{q}$ agree
on $\mathbb{A}^{c}$, we have not only $u$, but also $c_{{\cal V}}$ the same for all these preferences. Let $c_{{\cal P}}(\tau)=\sup_{B\in\mathbb{A}}[U(x_{B}(\tau))-U(B)]$. We have $U(x_{A}(\tau))-c_{{\cal P}}(\tau)\leq U(A)$ by definition. By Axiom 8, there exists $\tau^{*}\in\Delta(\cal{P})$ with $\frac{1}{2}\{x_A(\tau^*)\}+\frac{1}{2}\{x_B\}\succsim
\frac{1}{2}\{x_A\}+\frac{1}{2}\{x_B(\tau^*)\}$ for all $B$. Since $U$ represents $\succsim$, we have $U(\frac{1}{2}\{x_A(\tau^*)\}+\frac{1}{2}\{x_B\})\geq U(\frac{1}{2}\{x_A\}+\frac{1}{2}\{x_B(\tau^*)\})$. Since $U$  is linear over singleton menus, we have $\frac{1}{2}U(\{x_A(\tau^*)\})+\frac{1}{2}U(\{x_B\})\geq \frac{1}{2}U(\{x_A\})+\frac{1}{2}U(\{x_B(\tau^*)\})$. Rearranging the above terms and multiplying both sides by $2$, we obtain $U(x_{A}(\tau^{*}))-U(A)\geq U(x_{B}(\tau^{*}))-U(B)$ for all $B$. Thus, $U(x_{A}(\tau^{*}))-U(A)\geq c_{{\cal P}}(\tau^{*})$, and so $U(A)=\sup_{\tau\in\Delta(\cal{P})}[U(x_{A}(\tau))-c_{{\cal P}}(\tau)]$.

Now let $C=\{f\}$ and $D=\{g\}$ be two singleton menus. Since $U$ is linear over singleton menus, we have $U(x_C(\tau^*))=u(x_C(\delta_{p_{{\tau^*}}}))$ and $U(x_D(\tau^*))=u(x_D(\delta_{p_{\tau^*}}))$ where $p_{\tau^*}=\int p.\tau^*(dp)$. Given that we have $\{x_C(\tau^*)\}\,1/2\,\{x_D\}\sim
\{x_C\}\,1/2\,\{x_D(\tau^*)\}$, by the linearity of $U$ we obtain $U(\{x_C\})-U(\{x_C(\tau^*)\})=U(\{x_D\})-U(\{x_D(\tau^*)\})$, and so $u(f)\cdot(p_0-p_{\tau^*})=u(g)\cdot(p_0-p_{\tau^*})$. Since $f$ and $g$ are arbitrary acts, and so they can be taken such that $u(f)\neq u(g)$, we must have $p_0=p_{\tau^*}$ showing that $\tau^* \in \Delta_0(\cal P)$. With finite state space $S$, $\Delta({\cal P})$ is compact; since $\Delta_0({\cal P})$ is a closed subset of $\Delta({\cal P})$, it is compact; the function
$U(x_{A}(\tau))$ is continuous, and $c_{{\cal P}}(\tau)$ is l.s.c.  in $\tau$ as established in Theorem \ref{del}. Since $\Delta_0({\cal P})$ is compact and $U(x_{A}(\tau))-c_{{\cal P}}(\tau)$ is u.s.c., the supremum is attained. Using the definition of $x_{A}(\tau)$, we derive the sequential model, 
\[
U(A)=\max_{\tau\in\Delta_0({\cal P})}\left[\int_{{\cal P}}\left(\max_{\lambda\in\Delta(\mathcal{V})}[b_{A^{p}}^{u}(\lambda)-c_{{\cal V}}(\lambda)]\right)\tau(dp)\,-\,c_{{\cal P}}(\tau)\right],
\]
completing the sufficiency part of the proof. \qed

\paragraph*{Proof of Theorem \ref{ident}:}
Since $\succsim$ is a costly delegation preference, it has a canonical
representation $(u,p_{0},c^{*},\hat{\lambda})$ where $c^{*}(\pi)=\sup_{A\in\mathbb{A}}\left(b_{A}^{u}(\pi)-U(A)\right)$. By definition, $c^{*}$ is the unique minimal cost function
representing $\succsim$. The binary relation $\succsim^{c}$ restricted
to the space of constant menus $\mathbb{A}^{c}$ satisfies Axioms
1-6 in \cite{MO18}. As such, $U$ as the representation
of $\succsim$ satisfies $U(A)=\max_{\lambda\in\Delta(\mathcal{V})}[b_{A}^{u}(\lambda)-c_{{\cal V}}^{*}(\lambda)]$
for all $A\in\mathbb{A}^{c}$, where the cost function over tastes
is defined by the variational formula $c_{{\cal V}}^{*}(\lambda)=\sup_{A\in\mathbb{A}^{c}}[b_{A}^{u}(\lambda)-U(\{x_{A}\})]$
for all $\lambda\in\Delta({\cal V})$. By Theorem 2 in \cite{MO18}, we know that $c_{{\cal V}}^{*}$ is the minimal cost function
providing a representation (together with the utility function $u$)
for the binary relation $\succsim^{c}$. Moreover, by Corollary 1
in \cite{MO18}, the minimal cost function $c_{{\cal V}}^{*}$
is grounded, lower semi-continuous (l.s.c.), convex, and monotone (with
respect to the order in stochastic conflicts; see Definition 2 in
\cite{MO18}).

Note that since $\succsim$ is a sequential costly Bayesian persuasion
preference, we have $c^{*}(\pi)=c_{{\cal P}}^{*}(\pi_{{\cal P}})+\int_{{\cal P}}c_{{\cal V}}^{*}(\pi_{p})\pi_{{\cal P}}(dp)$. Now let $\hat{c}_{{\cal P}}(\tau)=\sup_{B\in\mathbb{A}}(U(x_{B}(\tau))-U(B))$
for any $\tau\in\Delta_{0}({\cal P})$. By the proof of Theorem \ref{sequential},
we have $U(A)=\sup_{\tau}[U(x_{A}(\tau))-\hat{c}_{{\cal P}}(\tau)]$ for any $A\in\mathbb{A}$.
Let $\tau_{A}\in\Delta_{0}({\cal P})$ be
such that $U(A)=U(x_{A}(\tau_{A}))-\hat{c}_{{\cal P}}(\tau_{A})$.
Using the definition of $x_{A}(\tau_{A})$ and the fact that $U$
is linear over singleton menus, we have 
\[
U(x_{A}(\tau_{A}))=\int_{{\cal P}}U(x_{A^{p}})\tau_{A}(dp)=\int_{{\cal P}}[\max_{\lambda\in\Delta(\mathcal{V})}[b_{A^{p}}^{u}(\lambda)-c_{{\cal V}}^{*}(\lambda)]]\tau_{A}(dp).
\]
Thus, we derive $\hat{c}_{{\cal P}}(\tau_{A})+\int_{{\cal P}}c_{{\cal V}}^{*}(\lambda_{A^{p}})\tau_{A}(dp)=\int_{{\cal P}}b_{A^{p}}^{u}(\lambda_{A^{p}})\tau_{A}(dp)-U(A)$
where we use the identity $U(x_{A^{p}})=b_{A^{p}}^{u}(\lambda_{A^{p}})-c_{{\cal V}}^{*}(\lambda_{A^{p}})$ for each $p\in{\cal P}$.
Let $\pi(A)\in\Delta_{0}({\cal P}\times{\cal V})$ be a joint distribution
such that its marginal over ${\cal P}$ is equal to $\tau_{A}$, and
its conditional distribution for each $p\in{\cal P}$ be equal to
$\lambda_{A^{p}}$; that is, $\pi(A)_{{\cal P}}=\tau_{A}$ and $\pi(A)_{p}=\lambda_{A^{p}}$
for each $p\in{\cal P}$. Thus, we must have
\begin{align*}
\hat{c}_{{\cal P}}(\tau_{A})+\int_{{\cal P}}c_{{\cal V}}^{*}(\lambda_{A^{p}})\tau_{A}(dp) 
 & =b_{A}^{u}(\pi(A))-U(A)\\
 & \leq\sup_{B\in\mathbb{A}}\left(b_{B}^{u}(\pi(A))-U(B)\right)  =c^{*}(\pi(A))\\
& =c_{{\cal P}}^{*}(\pi(A)_{{\cal P}})+\int_{{\cal P}}c_{{\cal V}}^{*}(\pi(A)_{p})\pi_{{\cal P}}(dp)\\
 & =c_{{\cal P}}^{*}(\tau_{A})+\int_{{\cal P}}c_{{\cal V}}^{*}(\lambda_{A^{p}})\pi_{{\cal P}}(dp).
\end{align*}
Since $c^{*}$ is the minimal cost function representing $\succsim$,
we have $\hat{c}_{{\cal P}}(\tau_{A})=c_{{\cal P}}^{*}(\tau_{A})$.
Since $A\in\mathbb{A}$ is arbitrary, we must have $\hat{c}_{{\cal P}}=c_{{\cal P}}^{*}$, and so $c_{{\cal P}}^{*}(\tau)=\sup_{A\in\mathbb{A}}(U(\{x_{A}(\tau)\})-U(\{x_{A}\}))$
and $c_{{\cal V}}^{*}(\lambda)=\sup_{A\in\mathbb{A}^{c}}[b_{A}^{u}(\lambda)-U(\{x_{A}\})]$
are the unique minimal cost functions representing $\succsim$ together
with $(u,p_{0})$. Note that since we
have $U(\{x_{A}(\tau)\})=U(\{x_{A}\})$ for any constant menu $A\in\mathbb{A}^c$ and for any $\tau$, the cost
function $c_{{\cal P}}^{*}$ is grounded. Moreover, $c_{{\cal P}}^{*}$
is the supremum of affine functions, and so it is l.s.c. and convex.
Finally, we claim that $c_{{\cal P}}^{*}$ is monotone in the Blackwell
order (i.e., the convex order). To see this, note that for any fixed
menu $A$, $U(x_{A^{p}})$ is a convex function of the posteriors
$p$ since we have $U(x_{A^{p}})=\max_{\lambda\in\Delta(\mathcal{V})}[b_{A^{p}}^{u}(\lambda)-c_{{\cal V}}^{*}(\lambda)]$.
Thus, for any mean preserving spread $\tau'$ of $\tau$, we have
$U(\{x_{A}(\tau')\})\geq U(\{x_{A}(\tau)\})$ implying that $c_{{\cal P}}^{*}(\tau')\geq c_{{\cal P}}^{*}(\tau)$with
strict inequality if $\tau'\neq\tau$.\qed

\paragraph*{Proof of Theorem \ref{comp}:}

The proof follows from Theorem 3 in \cite{MO18} and Corollary 4 in \cite{dDMO17}. \qed

\paragraph*{Proof of Proposition 1:}

By Theorem 2, $\succsim$ is represented by the Bayesian persuasion with uncertain bias model with parameters $(u,p_0,\Gamma,\lambda)$.
Then, $\succsim$ restricted to $\mathbb{A}^c$ has a random Strotz representation
$U(A)=\int_{\mathcal{V}}b^{u,v}_A \lambda(dv)$ for any $A\in\mathbb{A}^c$, where $b_{A}^{u,v}=\max_{f\in m_{v}(A)}u(f)$ and
$m_{v}(A)=\mathrm{argmax}_{f\in A}v(f)$.
Consider the following result, which we use for the proof of the proposition.

\begin{lemma} \label{Prop: M} 
Let $\succsim $ has a random Strotz representation over $\mathbb{A}^c$. 
Then $\succsim $ satisfies Axiom 11' (Stovall Finiteness) if and only if 
$\left\vert \mathrm{supp}(\lambda )\right\vert <N$.
\end{lemma}

\begin{proof}
\emph{[Necessity]:} Fix a random Strotz representation of $\succsim $ such that 
$\left\vert \mathrm{supp}(\lambda )\right\vert <N$. 
For any $A\in \mathbb{A}^{c}$ and $v\in \mathrm{supp}(\lambda )$, 
let $a_{A}^{v}\in \arg\max_{a\in A}v(a)$. 
Define $A^{\ast }\equiv \cup _{v\in \mathrm{supp}(\lambda )}\{a_{A}^{v}\}$. 
Since $\left\vert \mathrm{supp}(\lambda)\right\vert <N$, we have
$\left\vert A^{\ast }\right\vert <N$. 
We now show that $A^{\ast }$ is critical for $A$. 
Take any $B$ such that $A^{\ast }\subset \mathrm{co}(B)\subset \mathrm{co}(A)$. 
By the definition of $A^{\ast }$, we have
$\arg \max_{b\in B}v(b)=\arg \max_{a\in A^{\ast }}v(a)$ 
for all $v\in \mathrm{supp}(\lambda )$. 
Hence, $b_{A^{\ast }}^{u,v}=b_{B}^{u,v}$ for all $v\in \mathrm{supp}(\lambda )$, 
which implies that $A^{\ast }\sim B$.

\emph{[Sufficiency]:} Fix a random Strotz representation of $\succsim $ over $\mathbb{A}^c$. 
Let $E\subset \mathbb{A}^c$ be a sphere with full-support; that is, let $E$ be the surface of a closed ball. 
By Axiom 11', $E$ has a critical subset $E^{\ast }$ 
such that $\left\vert E^{\ast }\right\vert <N$.
We adapt the construction in \citet[p.958]{DLR09} to our setting. 
Since $E$ is a sphere, there is a one-to-one mapping $\kappa $ 
from $E$ to $\mathcal{V}$ where $\kappa (a)$ is 
a utility $v\in \cal V$ such that $a$ is the unique maximizer of $b\cdot v$ over $b\in E$. 
This means that $a$ is chosen by $v$ from $E$. 
Let $\mathcal{V}^{\ast }=\kappa (E^{\ast })
=\left\{ v\in \mathcal{V}:\kappa (a)=v
\text{ for some }a\in E^{\ast }\right\}$. Since $\left\vert E^{\ast }\right\vert <N$, we have
$\left\vert \mathcal{V}^{\ast}\right\vert <N$.

We want to show that \textrm{supp}$(\lambda )\subset \mathcal{V}^{\ast }$, and so 
we can conclude $\left\vert \mathrm{supp}(\lambda )\right\vert <N$. 
Suppose, for contradiction, that there is 
$v^{\prime }\in \textrm{supp}(\lambda )\setminus \mathcal{V}^{\ast }$. 
Let $a^{\prime }\in E$ be such that $\kappa (a^{\prime })=v^{\prime }$. Comparing $a^{\prime}$ against the choices in  $E^*$, we can decompose $\cal V$ into three sets, $V_+, V_0, V_{-}$ where $V_+=\{v\in\mathcal{V}:v(a')> \max_{a\in E^*}v(a)\}$, $V_0=\{v\in\mathcal{V}:v(a')= \max_{a\in E^*}v(a)\}$ and $V_{-}=\{v\in\mathcal{V}:v(a')< \max_{a\in E^*}v(a)\}$. By assumption, $v'\in V_+$; since $v'$ is in the support, $\lambda(V_+)>0$. We have $U(E^*\cup\{a'\})=\lambda(W)u(a')
    +\int_{\cal V \setminus W}b_{E^{\ast }}^{u,v}\lambda(dv)$ where $W=V_+\cup V_0$ if $u(a')\geq \max_{a\in E^*}u(a)$ and $W=V_0 \cup V_{-}$ if $u(a')< \max_{a\in E^*}u(a)$. Thus, $E^{\ast }\cup \{a^{\prime }\}\sim E^{\ast }$ if and only if $\lambda(W)u(a') = \int_{W}b_{E^{\ast }}^{u,v}\lambda(dv)$. In general, this equality cannot hold since we can always slightly perturb $a'$ to another $a''\in E$. Thus, we have $E^{\ast }\cup \{a^{\prime }\}\nsim E^{\ast } \sim E$ which violates Axiom 11', a contradiction. Hence, $\left\vert \mathrm{supp}(\lambda )\right\vert <N$.
\end{proof}

By Lemma \ref{Prop: M} given above, we obtain the proof of Proposition \ref{Prop: N}. \qed

\paragraph*{Proof of Proposition 2:} The proof follows from the proof of Lemma \ref{Prop: M} once we note that \textrm{supp}$(\lambda )\subset \mathcal{V}^{\ast }$ implies a finite support in either direction. \qed

\paragraph*{Proof of Proposition 3:}
By Theorem \ref{costly}, $\succsim$ is represented by the costly persuasion model with parameters $(u,p_0,c_{\cal{P}},\lambda)$. Suppose, for contradiction, that the support of the distribution $\lambda$ has more than two distinct utilities in $\cal{V}$. Let $a,b\in \Delta(X)$ such that $u(a)> u(b)$ and $\lambda_a=\lambda(\,\{v\in \mathcal{V}:v(a)\geq v(b)\}\,)>0$, $\lambda_b = \lambda(\,\{v\in\mathcal{V}:v(a)<v(b)\}\,)>0$. By construction, we have $U(\{a,b\}) = \lambda_a u(a) + \lambda_b u(b)$, and so $\{a,b\} \nsim \{a\}$ and $\{a,b\} \nsim \{b\}$ which contradicts Axiom 11. \qed

\paragraph*{Proof of Corollary 1:}
By Proposition 3, $\succsim$ is represented by the costly persuasion model with parameters $(u,p_0,c_{\cal{P}},v)$. Since Axiom $11'''$ holds, for any $a,b\in \Delta(X)$ such that $u(a)\geq u(b)$, we must have  $\{a,b\} \sim \{a\}$, and so  $v(a)\geq v(b)$ which implies that $u=v$. \qed

\paragraph*{Proof of Proposition 4:}
By Theorem \ref{sequential}, $\succsim$ is represented by the sequential persuasion model with parameters $(u,p_0,c_{\cal{P}},c_{\cal{V}})$. By Axiom 8' and the proof of Theorem  \ref{sequential}, for any $A \in \mathbb{A}$ we have $U(A)=U(x_A(\tau))-c_{\cal{P}}(\tau)$. Since $c_{\cal{P}}$ can be taken as grounded, we have $c_{\cal{P}}(\tau)=0$, and so  $\succsim$ is represented by sequential persuasion model with parameters  $(u,p_0,\tau,c_{\cal{V}})$ where $\tau \in \Delta_0(\cal P)$. \qed

\paragraph*{Proof of Corollary 2:}
By Proposition 4, $\succsim$ is represented by the sequential persuasion model with parameters $(u,p_0,\tau,c_{\cal{V}})$ with $\tau \in \Delta_0(\cal P)$. By Axiom 8'', $\tau$ can be taken as $\delta_{p_0}$. \qed

\bibliography{references}
\bibliographystyle{ecta}

\end{document}